\documentclass{ws-procs9x6}

\def\bb{\begin{eqnarray}}
\def\ee{\end{eqnarray}}

\def\nn{\nonumber}

\begin{document}

\title{Periodic Benjamin-Ono equation with discrete Laplacian
and 2D-Toda Hierarchy}

\author{Jun'ichi Shiraishi}

\address{
Graduate School of 
Mathematical Sciences, The University of Tokyo,\\
3-8-1 Komaba Meguro-ku Tokyo 153-8914, Japan\\
E-mail: shiraish@ms.u-tokyo.ac.jp}

\author{Yohei Tutiya}

\address{
Kanagawa Institute of Technology,\\
1030 Shimo-Ogino Atsugi-city Kanagawa 243-0292, Japan\\
E-mail: tutiya@gen.kanagawa-it.ac.jp}

\begin{abstract}
We study the
relation between 
the periodic Benjamin-Ono equation with 
discrete Laplacian and the two dimensional Toda hierarchy. 
We introduce the tau-functions $\tau_\pm(z)$ for the periodic Benjamin-Ono equation, 
construct two families of integrals of motion $\{M_1,M_2,\cdots\}$, 
$\{\overline{M}_1,\overline{M}_2,\cdots\}$,
and calculate some examples of the bilinear equations 
using the Hamiltonian structure.
We confirmed that 
some of the low lying bilinear equations agree with the ones obtained from 
a certain reduction of the 2D Toda hierarchy. 
\end{abstract}

\keywords{Benjamin-Ono equation, 2D Toda hierarchy}

\bodymatter

\section{Introduction}
\subsection{Periodic Benjamin-Ono equation with discrete Laplacian}

Let $\gamma$ be a complex parameter 
satisfying ${\rm Im}(\gamma)\geq 0$.
Let $x$ and $t$ be real independent variables, and $\eta(x,t)$  
be an analytic function satisfying the periodicity condition $\eta(x+1,t)=\eta(x,t)$.
In Refs.~\refcite{ST} and \refcite{TS},
we considered the integro-differential equation
\begin{eqnarray}
\frac\partial{\partial t}\eta(x,t)
=
\eta(x,t)\cdot 
\frac{i }{2 }
\int_{-1/2}^{1/2}
\hspace{-2.2em}\backslash\hspace{.5em}\,\,\,\,\,
\Bigl(\Delta_\gamma \cot(\pi(y-x)) \Bigl) 
\eta(y,t)dy,
\label{PBO}
\end{eqnarray}
where the discrete Laplacian $\Delta_\gamma$ is defined by
$(\Delta_\gamma f)(x)=f(x-\gamma)-2f(x)+f(x+\gamma)$,
and the integral $\int\hspace{-0.8em}\backslash\hspace{.5em}$
means the Cauchy principal value. 
This can be regarded as a periodic version of the 
Benjamin-Ono \cite{B,O} equation associated with 
the discrete Laplacian $\Delta_\gamma$.

For the sake of simplicity, set $z=e^{2\pi i x}$ and $q=e^{2\pi i \gamma}$. 
By abuse of notation, 
we use the notation $\eta(z)$
to indicate the dependence on $z$. 
{}From the spatial periodicity, 
we have the Fourier series 
expansion $\eta(z)=\sum_{n\in {\bf Z}}\eta_{n}z^{-n}$.
Set $\eta_+(z)=\sum_{n>0}\eta_{-n}z^{n}$
and 
$\eta_-(z)=\sum_{n>0}\eta_{n}z^{-n}$. 
Note that we have $\eta(z)=\eta_+(z)+\eta_0+\eta_-(z)$.
Then the equation (\ref{PBO})
can be expressed as
\bb
&&\partial_{t} \eta(z)
=\eta(z)\sum_{l\neq0}{\rm sgn}(l)(1-q^{|l|})\eta_{-l}z^l\nonumber \\
&&\qquad=
\eta(z)(\eta_+(z) -\eta_+(zq) -\eta_-(z) +\eta_-(z/q)).
\label{BO}\ee

\subsection{Poisson algrbra and Toda field equation}
We show that one can introduce another time $\overline{t}$
and obtain the 2D Toda field equation \cite{UT}
by using the Poisson Heisenberg algebra for our periodic Benjamin-Ono equation 
with discrete Laplacian \cite{ST,TS}.
As for the Hamiltonian structure for the usual Benjamin-Ono equation, see 
Refs. \refcite{Sa} and \refcite{LR}. See Ref. \refcite{ABW} also.

Our Poisson algebra 
is generated by $\alpha_n$ ($n\in {\bf Z}_{\neq 0}$) with the Poisson brackets
\bb
\{\alpha_n,\alpha_m\}={\rm sgn}(n)(1-q^{|n|})\delta_{n+m,0},\label{alpha}
\ee
where ${\rm sgn}(n)=|n|/n$ for $n\neq 0$ and   ${\rm sgn}(0)=0$.
\begin{definition}
Set 
\bb
\tau_+(z)=\exp\left(-\sum_{n>0}\frac{\alpha_{-n}}{1-q^n}z^{n}\right),\quad
\tau_-(z)=\exp\left(-\sum_{n>0}\frac{\alpha_n}{1-q^n}z^{-n}\right).
\ee
We call $\tau_\pm(z)$ the tau-functions.
\end{definition}

Express the dependent variable $\eta(z)$
in terms of the tau-functions and a constant $\varepsilon$ as
\bb
&&\eta(z)=\sum_{n\in {\bf Z}}\eta_n z^{-n}=\varepsilon\exp\left(\sum_{n\neq0}\alpha_nz^{-n}\right)
=\varepsilon\frac{\tau_-(z/q)\tau_+(zq)}{\tau_-(z)\tau_+(z)}.\label{eta}
\ee
We need to  introduce another dependent variables $\xi(z)$ by
\bb
\xi(z)&=&\sum_{n\in {\bf Z}}\xi_n z^{-n}=
{1\over  \varepsilon}\exp\left(-\sum_{n\neq0}\alpha_nq^{-|n|/2}z^{-n}\right)\nn\\
&=&{1\over \varepsilon}\frac{\tau_-(zq^{1/2})\tau_+(zq^{-1/2})}{\tau_-(zq^{-1/2})\tau_+(zq^{1/2})}.\label{xi}
\ee

The Poisson brackets among our dependent variables 
can be calculated as follows.
\begin{lemma}\label{Poisson-Lemma}
We have 
\bb
&&\{\eta(z), \eta(w)\}=
\eta(z)\eta(w)\sum_{l\neq0}{\rm sgn}(l)(1-q^{|l|})\left(\frac{w}{z}\right)^{l}, \label{eta-eta}\\
&&\{\xi(z), \xi(w)\}=
\xi(z)\xi(w)\sum_{l\neq0}{\rm sgn}(l)(q^{-|l|}-1)\left(\frac{w}{z}\right)^{l}, \label{xi-xi}\\
&&\{\eta(z), \xi(w)\}=
\delta(q^{1/2}w/z){\tau_+(zq)\tau_+(z/q)\over \tau_+(z) \tau_+(z)}\nonumber\\
&&\qquad\qquad\qquad-
\delta(q^{-1/2}w/z){\tau_-(zq)\tau_-(z/q)\over \tau_-(z) \tau_-(z)},\label{eta-xi}
 \\
%
&&\{\eta(w),\tau_-(z)\}
=\eta(w)\tau_-(z)\sum_{n>0}\left(\frac{w}{z}\right)^{n},\label{eta-tau-}\\
%
&&\{\eta(w),\tau_+(z)\}
=-\eta(w)\tau_+(z)
\sum_{n>0}\left(\frac{z}{w}\right)^{n},\label{eta-tau+}\\
&&\{\xi(w),\tau_-(z)\}
=-\xi(w)\tau_-(z)\sum_{n>0}q^{-n/2}\left(\frac{w}{z}\right)^{n},\label{xi-tau-}\\
%
&&\{\xi(w),\tau_+(z)\}
=\xi(w)\tau_+(z)
\sum_{n>0}q^{-n/2}\left(\frac{z}{w}\right)^{n}.\label{xi-tau+}
\ee
where $\delta(z)=\sum_{n\in{\bf Z}}z^n$.
\end{lemma}

\begin{remark}
In Ref. \refcite{FHHSY}, a deep connection was found 
between the 
Macdonald Polynomials \cite{Mac} $P_\lambda(x;q,t)$ and the level one representation of the 
quantum algebra of Ding-Iohara ${\cal U}(q,t)$. 
We note that the $\alpha_n$, $\eta(z)$ and $\xi(z)$ in the present paper 
are the level one generators of
the Ding-Iohara algebra in the classical (namely commutative) limit given by 
letting the parameter as
$t\rightarrow 1$. (Note that $t$ here is one of the two parameters $q,t$ 
for the Macdonald polynomials and shold not be confused with the 
time $t$.)
\end{remark}

\begin{proposition}
We have 
$\partial_t  \eta(z)=\{\eta_0,\eta(z)\}$.
Hence we identify $\eta_0$ with our Hamiltonian corresponding to the time $t$.
\end{proposition}
\begin{proof}
This follows from (\ref{BO}) and (\ref{eta-eta}) in Lamma \ref{Poisson-Lemma}.
\end{proof}

\begin{proposition}
We have 
$\{\eta_0,\xi_0\}=0$.
\end{proposition}
\begin{proof}
This follows from (\ref{eta-xi}) in Lamma \ref{Poisson-Lemma}.
\end{proof}

\begin{remark}
Since our Poisson algebra is  the classical 
limit of the 
Ding-Iohara algebra, we have two sets of mutually Poisson commutative families,
having $\eta_0$ and $\xi_0$ respectively, 
in the same way as was discussed in Ref. \refcite{FHHSY}.
As for the explicit form of them, see (\ref{I_k}) and (\ref{Ibar_k}) below.
\end{remark}

Because of the commutativity $\{\eta_0,\xi_0\}=0$, 
we may interpret  $\xi_0$ as another Hamiltonian 
corresponding to time $\overline{t}$. 

\begin{definition}
Define $\partial_{\overline{t}}*= \{*,\xi_0\}$.
\end{definition}

\begin{proposition}
{}From (\ref{eta-xi}) we have 
\begin{eqnarray}
\partial_{\overline{t}}\eta(z)=
{\tau_+(zq)\tau_+(z/q)\over \tau_+(z) \tau_+(z)}-
{\tau_-(zq)\tau_-(z/q)\over \tau_-(z) \tau_-(z)}, \label{eta-tb}
\end{eqnarray}
\end{proposition}

Now we turn to the 
equations for the tau-functions $\tau_\pm(z) $.
\begin{proposition}
{}From (\ref{eta-tau-}), (\ref{eta-tau+}),  (\ref{xi-tau-}) 
and (\ref{xi-tau+}), we have
\begin{eqnarray}
&&\partial_t \tau_-(z)=\eta_-(z)\tau_-(z),
\qquad 
\partial_t \tau_+(z)=-\eta_+(z)\tau_-(z), \label{partial-tau}\\
&&\partial_{\overline{ t}} \tau_-(zq^{-1/2})=\xi_-(z)\tau_-(zq^{-1/2}),
\qquad 
\partial_{\overline{t} }\tau_+(zq^{1/2})=-\xi_+(z)\tau_-(zq^{1/2}),\nn
\end{eqnarray}
where 
$\xi_+(z)=\sum_{n>0}\xi_{-n}z^{n}$
and 
$\xi_-(z)=\sum_{n>0}\xi_{n}z^{-n}$. 
\end{proposition}

Suitable combinations of these may give us 
equations written in terms of the Hirota derivatives.

\begin{definition}
Define the Hirota derivative $D_{t_1},D_{t_2},\cdots$ by
\begin{eqnarray}
&&\left( D_{t_1}^{k_1}D_{t_2}^{k_2}\cdots \right)f\cdot g\\
&&=
\partial_{a_1}^{k_1} \partial_{a_2}^{k_2}\cdots f(t_1+a_1,t_2+a_2, \cdots)g(t_1-a_1,t_2-a_2, \cdots)
\Biggl|_{a_1=a_2=\cdots=0}.\nn
\end{eqnarray}
\end{definition}

\begin{proposition}
We have the Hirota equations
\begin{eqnarray}
&&D_t \tau_-(z) \cdot \tau_+(z)=
 \varepsilon \tau_-(z/q) \tau_+(zq)-\eta_0 \tau_-(z)\tau_+(z),\label{Hirota-t}\\
&& D_{\overline{t}}\tau_-(zq^{-1/2}) \cdot \tau_+(zq^{1/2})\label{Hirota-tb}\\
&&\qquad =
 \varepsilon^{-1} \tau_-(zq^{1/2}) \tau_+(zq^{-1/2})-\xi_0 \tau_-(zq^{-1/2})\tau_+(zq^{1/2}),\nn\\
 &&
{1\over 2} D_t D_{\overline{t}}  \tau_\pm (z) \cdot  \tau_\pm(z) \label{Toda}\\
 &&
\qquad  +{~}
  \tau_\pm (z q)  \cdot  \tau_\pm (z/q)- \tau_\pm (z) \cdot  \tau_\pm (z) =0.  \nn
\end{eqnarray}
\end{proposition}
\begin{proof}
{}From (\ref{partial-tau}), we have 
$\partial_t\tau_-(z) \cdot \tau_+(z)-
 \tau_-(z) \cdot \partial_t \tau_+(z)=(\eta(z)-\eta_0) \tau_-(z)\tau_+(z)$.
Using (\ref{eta}) we have (\ref{Hirota-t}). The equation (\ref{Hirota-tb}) 
can be derived in the same way.
Eq. (\ref{Toda})
is obtained from (\ref{eta-tb}) and (\ref{partial-tau}).
\end{proof}

\begin{remark}
Note that Eq.  (\ref{Toda}) is nothing but the Toda field equation 
written in terms of the tau-function \cite{UT}.
\end{remark}

One finds that the Heisenberg generators 
correspond to the standard dependent variables 
of the Toda field theory .
\begin{definition}
Set
\begin{eqnarray}
\phi_+(z)=\sum_{n>0}\alpha_{-n}z^n,
\qquad 
\phi_-(z)=-\sum_{n>0}\alpha_{n}z^{-n}.
\end{eqnarray}
\end{definition}

\begin{proposition}
The $\phi_\pm (z)$ satisfy the Toda field equation
\begin{eqnarray}
\partial_t \partial_{\overline{t}}
\phi_\pm(z)
=
e^{\phi_\pm(z)-\phi_\pm(z/q)}-
e^{\phi_\pm(zq)-\phi_\pm(z)}.\label{TodaToda}
\end{eqnarray}
\end{proposition}
\begin{proof}
This follows from Eqs. (\ref{alpha}) and (\ref{eta-xi}).
\end{proof}

Motivated by the appearance  of the Toda field equation  (\ref{TodaToda}),
in this article we will try to understand how 
the 2D Toda hierarchy appears from the point of view of 
the Hamiltonian structure.

The present paper is organized as follows.
In Section 2, we recall the Hirota-Miwa equation\cite{Hi,Mi} (\ref{HM}) 
for the 2D Toda hierarchy.
Based on the $n$-soliton solutions,
we derive some variants of the bilinear equations (Proposition \ref{3HM}).
In Section 3, two sets of integrals of motion
$M_1,M_2,\cdots$ and $\overline{M}_1,\overline{M}_2,\cdots$
are introduced (Definition \ref{MMbar}). 
Since at present we lack enough technologies
to handle the evolution equations in general, 
we need to restrict ourself to some low lying cases. 
To show some evidences of the agreement, 
we check up to certain degree that exactly the same 
equations are obtained both from  
Proposition \ref{3HM} and the Hamiltonian $M_k$'s.
Our observation is summarized in Conjecture \ref{maincon}.

\section{Hirota-Miwa equation for 2D Toda hierarchy}
\subsection{Hirota-Miwa equation}
We briefly recall the Hirota-Miwa equation \cite{Hi,Mi} for the 2D Toda
hierarchy \cite{UT}. Using the $n$-soliton solution,
we derive several variants of bilinear equations
which can be connected to Eqs. (\ref{Hirota-t}), (\ref{Hirota-tb}) and
(\ref{Toda}) obtained in the 
previous section.

Let $s$, $t=(t_1,t_2,\cdots)$ and $\overline{t}=(\overline{t}_1,\overline{t}_2,\cdots)$
be independent variables, and 
$\tau(s,t,\overline{t})$ be the tau-function of the
2D Toda hierarchy.
For a parameter $\lambda$,
we use the standard notation for the infinite vector
$[\lambda]=(\lambda,{1\over 2}\lambda^2,{1\over 3}\lambda^3,\cdots)$.

The Hirota-Miwa equation for the 2D Toda hierarchy is written as follows.
\begin{eqnarray}
&&(1-\alpha\beta)\tau(s,t,\bar{t})\tau(s,t+[\alpha],\bar{t}+[\beta])-
\tau(s,t+[\alpha],\bar{t})\tau(s,t,\bar{t}+[\beta]) \nn\\
&&
+~\alpha \beta
\tau(s+1,t+[\alpha],\bar{t})\tau(s-1,t,\bar{t}+[\beta])=0.\label{HM}
\end{eqnarray}

It is well known that the $n$-soliton solution to the Hirota-Miwa equation 
is given by
\begin{eqnarray}
&&\tau(s,t,\overline{t})=
\sum_{r=0}^n
 \sum_{I\subset \{1,2,\cdots,n\}\atop |I|=r }
\prod_{\{i,j\}\subset I \atop
i<j} {(\lambda_i-\lambda_j)(\mu_i-\mu_j)\over 
(\lambda_i-\mu_j)(\mu_i-\lambda_j)}\\
&&\qquad\qquad\times
\prod_{k\in I} 
(\lambda_k/\mu_k)^s
e^{\sum_{i=1}^\infty (t_i \lambda_k^i+\overline{t}_i \lambda_k^{-i})-
\sum_{i=1}^\infty (t_i \mu_k^i+\overline{t}_i \mu_k^{-i})}.\nn
\end{eqnarray}

Let $a_1,\cdots,a_n$ be parameters.
Set $\lambda_k=a_k,\mu_k=q a_k$ for $k=1,2,\cdots,n$.
Write $z=q^{-s}$. Then we have $(\lambda_k/\mu_k)^s=q^{-s}=z$.

We define $\tau_+(z,t,\overline{t})$ by 
the $n$-soliton solution $\tau(s,t,\overline{t})$ of 2D Toda hierarchy 
under this specialization ($\lambda_k=a_k,\mu_k=q a_k$).
\begin{definition}
Set
\begin{eqnarray}
&&\tau_+(z,t,\overline{t})=
\sum_{r=0}^n
z^r \sum_{I\subset \{1,2,\cdots,n\}\atop |I|=r }
\prod_{\{i,j\}\subset I \atop
i<j} {(a_i-a_j)^2\over (a_i-qa_j)(a_i-q^{-1}a_j)}\label{tau-plus}\\
&&\qquad\qquad
 \times\prod_{k\in I} e^{
 \sum_{i=1}^\infty (1-q^i)t_ia_k^i + \sum_{i=1}^\infty (1-q^{-i})\overline{t}_ia_k^{-i} }.\nn
\end{eqnarray}
\end{definition}

Note that $\tau_+(z,t,\overline{t})$ is a polynomial in $z$ whose degree is $n$.

To introduce $\tau_-(z,t,\overline{t})$, we need a Lemma.
\begin{lemma}\label{tau+totau-}
We have
\begin{eqnarray*}
&&
\tau_+(z,t,\bar{t}-[\beta])\\
&&=
\sum_{r=0}^n
z^r \sum_{I\subset \{1,2,\cdots,n\}\atop |I|=r }
\prod_{\{i,j\}\subset I \atop
i<j} {(a_i-a_j)^2\over (a_i-qa_j)(a_i-q^{-1}a_j)}\\
&&\qquad\times
\prod_{k\in I} 
{1-\beta/a_k\over 1-\beta/qa_k}
e^{
 \sum_{i=1}^\infty (1-q^i)t_ia_k^i + \sum_{i=1}^\infty (1-q^{-i})\overline{t}_ia_k^{-i} }\\
&&= z^n
\prod_{1\leq i<j\leq n} {(a_i-a_j)^2\over (a_i-qa_j)(a_i-q^{-1}a_j)}\\
&&\qquad \times
\prod_{k=1}^n 
{1-\beta/a_k\over 1-\beta/qa_k}e^{
 \sum_{i=1}^\infty (1-q^i)t_ia_k^i + \sum_{i=1}^\infty (1-q^{-i})\overline{t}_ia_k^{-i} }\\
&&\qquad \times\sum_{r=0}^n
z^{-r} \sum_{I\subset \{1,2,\cdots,n\}\atop |I|=r }
\prod_{\{i,j\}\subset I \atop
i<j} {(a_i-a_j)^2\over (a_i-qa_j)(a_i-q^{-1}a_j)}\\
&&\qquad \times
\prod_{k\in I} d_k (\beta)
e^{-
 \sum_{i=1}^\infty (1-q^i)t_ia_k^i - \sum_{i=1}^\infty (1-q^{-i})\overline{t}_ia_k^{-i} },
 \end{eqnarray*}
 where
\begin{eqnarray*}
&&d_k(\beta)=
{1-\beta/q a_k\varepsilon \over 1-\beta/a_k}
\prod_{j\neq k\atop 1\leq j\leq n}
 { (a_k-qa_j)(a_k-q^{-1}a_j)\over (a_k-a_j)^2}.
\end{eqnarray*}
\end{lemma}

Now we define $\tau_-(z,t,\overline{t})$ by the following Laurent polynomial.
\begin{definition}
Set
\begin{eqnarray}
&&\tau_-(z,t,\overline{t})\nn\\
&&=\tau_+(z,t,\overline{t}-[q^n \varepsilon])
\times
z^{-n}
\prod_{1\leq i<j\leq n} {(a_i-qa_j)(a_i-q^{-1}a_j)\over (a_i-a_j)^2}\\
&&\qquad \times
\prod_{k=1}^n 
{1-q^{n-1}\varepsilon/a_k\over 1-q^n\varepsilon/a_k}e^{
- \sum_{i=1}^\infty (1-q^i)t_ia_k^i -\sum_{i=1}^\infty (1-q^{-i})\overline{t}_ia_k^{-i} }\nn\\
&&=
\sum_{r=0}^n
z^{-r} \sum_{I\subset \{1,2,\cdots,n\}\atop |I|=r }
\prod_{\{i,j\}\subset I \atop i<j} {(a_i-a_j)^2\over (a_i-qa_j)(a_i-q^{-1}a_j)}\label{tau-minus}\\
&&\qquad \times
\prod_{k\in I} d_k (q^n \varepsilon)e^{-
 \sum_{i=1}^\infty (1-q^i)t_ia_k^i - \sum_{i=1}^\infty (1-q^{-i})\overline{t}_ia_k^{-i} }.
\nonumber
\end{eqnarray}
\end{definition}

\begin{proposition}\label{3HM}
We have
\begin{eqnarray}
&&\tau_-(z,t+[\alpha],\overline{t})\tau_+(z,t,\overline{t})\nonumber\\
&&
=(1-\alpha q^n \varepsilon ) 
\prod_{k=1}^n
 { (1-\alpha a_k)\over (1-\alpha q a_k)}
 \tau_-(z,t,\overline{t})\tau_+(z,t+[\alpha],\overline{t})\label{HM+-1}\\
 &&\qquad +\alpha \varepsilon \,
 \tau_-(z/q,t+[\alpha],\overline{t})\tau_+(zq,t,\overline{t}),\nonumber\\
 &&\tau_-(z/q,t,\overline{t}+[\beta])\tau_+(z,t,\overline{t})\nonumber\\
&&
=(1-\beta/q^{n} \varepsilon) 
\prod_{k=1}^n
 { (1-\beta/a_k )\over (1-\beta/q a_k)}
 \tau_-(z/q,t,\overline{t})\tau_+(z,t,\overline{t}+[\beta])\label{HM+-2}\\
 &&\qquad +(\beta/ \varepsilon) \,
 \tau_-(z,t,\overline{t}+[\beta])\tau_+(z/q,t,\overline{t}),\nonumber\\
 &&\tau_\pm(z,t+[\alpha],\bar{t})\tau_\pm(z,t,\bar{t}+[\beta]) \label{HM3}\\
  &&=(1-\alpha\beta)\tau_\pm(z,t,\bar{t})\tau_\pm(z,t+[\alpha],\bar{t}+[\beta])\nn\\
  &&
\qquad +~\alpha \beta
\tau_\pm(z/q,t+[\alpha],\bar{t})\tau_\pm(zq,t,\bar{t}+[\beta])=0.\nn
\end{eqnarray}
\end{proposition}
\begin{proof}
Eq. (\ref{HM+-1}) follows from the Hirota-Miwa equation (\ref{HM}), (\ref{tau-plus}), (\ref{tau-minus}) and 
Lemma \ref{tau+totau-}.
Noting that we have $\tau(s,t-[\alpha],\overline{t})=\tau(s+1,t,\overline{t}-[\alpha^{-1}])$,
we have (\ref{HM+-2}) in the same way. 
Eq. (\ref{HM3}) follows from Eq. (\ref{HM}).
\end{proof}
\begin{remark}
Note taht we may write
\begin{eqnarray}
&&(1-\alpha q^n \varepsilon ) 
\prod_{k=1}^n
 { (1-\alpha a_k )\over (1-\alpha q a_k)}=
 \exp\left(-\sum_{i=1}^\infty M_i   \alpha^i \right),\\
 &&(1-\beta/ q^n \varepsilon ) 
\prod_{k=1}^n
 { (1-\beta/ a_k )\over (1-\beta/ q a_k)}=
 \exp\left(-\sum_{i=1}^\infty \overline{M}_i   \beta^i \right),
 \\
 &&M_i={1-q^i\over i}\left(a_1^i+\cdots +a_n^i+q^{ni}\varepsilon^i+
 q^{(n+1)i}\varepsilon^i+ q^{(n+2)i}\varepsilon^i+\cdots \right), \label{eigenM}\\
 &&\overline{M}_i={1-q^{-i}\over i}\left(a_1^{-i}+\cdots +a_n^{-i}+q^{-ni}\varepsilon^{-i}+
 q^{-(n+1)i}\varepsilon^{-i}+\cdots \right). \label{eigenMbar}
\end{eqnarray}
\end{remark}
\begin{proposition}
By expanding (\ref{HM+-1}) in $\alpha$, we have 
the Hirota equations
\begin{eqnarray}
&&(D_{t_1}+M_1)\tau_-(z)\cdot \tau_+(z)=
\varepsilon \tau_-(z/q)\cdot \tau_+(zq),\label{To-1}\\
&&
(D_{t_2}+2M_2)\tau_-(z)\cdot \tau_+(z)=
\varepsilon (D_{t_1}+M_1) \tau_-(z/q)\cdot \tau_+(zq),\label{To-2}\\
&&(D_{t_3}+3M_3)\tau_-(z)\cdot \tau_+(z)
+
{1\over 8}(D_{t_1}+M_1)^3\tau_-(z)\cdot \tau_+(z)\label{To-3}\\
&&\qquad=
{3\over 4}
\varepsilon (D_{t_2}+2M_2) \tau_-(z/q)\cdot \tau_+(zq)+
{3\over 8}
\varepsilon (D_{t_1}+M_1)^2 \tau_-(z/q)\cdot \tau_+(zq),\nn
\end{eqnarray}
and so on. Here $M_i$'s are defined in (\ref{eigenM}).
\end{proposition}

Thus we found that Eq. (\ref{Hirota-t}) coincide with 
Eq. (\ref{To-1}) under the identification $t=t_1,\overline{t}=\overline{t}_1$.
Eqs. (\ref{Hirota-tb}) and  (\ref{Toda}) coincide with the 
first nontrivial equation from Eqs. (\ref{HM+-2}) and (\ref{HM3}) respectively.

In the next section, we will check that Eqs. (\ref{To-2}) and Eq. (\ref{To-3})
also agree with the equations derived from the Poisson structure.

\section{Poisson algebra and 2D Toda hierarchy}
\subsection{elementary and power sum symmetric functions}
We need some facts about the symmetric functions\cite{Mac}. 
Let $x=(x_1,x_2,\cdots)$ be an infinite set of independent indeterminates.
Let $e_n(x)$ be the $n$-th elementary symmetric function, and $p_n(x)$
be the $n$-th power sum function.
The generating functions for them are given by
$E(y)=\sum_{n=0}^\infty e_n(x)y^n=\prod_{i=1}^\infty (1+x_i y)$,
and $P(y)=\sum_{n=1}^\infty {1\over n}p_n(x)y^n=-\log E(-y)$.
Solving the equation $P'(y)=-E'(-y)/E(-y)$, we have 
\begin{eqnarray}
p_n=\left| 
\begin{array}{lllll} 
e_1&1&0&\cdots&0\\
2e_2&e_1&1&\cdots&\\
3e_3&e_2&e_1&\cdots&\\
\vdots&&&\ddots&\vdots\\
ne_n&e_{n-1}&e_{n-2}&\cdots&e_1
\end{array}\right|.
\end{eqnarray}

\subsection{Integrals of motion from Ding-Iohara algebra}
First we introduce some notations.
We denote the constant term $f_0$ of a series 
$f(z)=\sum_{n\in {\bf Z}} f_n z^n$  by $[f(z)]_1$.
We also use the same symbol for the case of 
a series with several variables. 
For examples, 
by $[f(z_1,z_2)]_1$ we denote the constant term $f_{0,0}$ of the  series 
$f(z_1,z_2)=\sum_{n_1,n_2\in {\bf Z}} f_{n_1,n_2} z_1^{n_1}z_2^{n_2}$.

\begin{definition}
Define the integrals of motion by 
\begin{eqnarray}
&&I_k=\left[ \prod_{1\leq i<j\leq  k}{1-w_j/w_i\over 1-qw_j/w_i} \eta(w_1)\eta(w_2)\cdots \eta(w_k)\right]_1,\label{I_k}\\
&&\overline{I}_k=\left[ \prod_{1\leq i<j\leq  k}{1-w_j/w_i\over 1-q^{-1}w_j/w_i} \xi(w_1)\xi(w_2)\cdots \xi(w_k)\right]_1,\label{Ibar_k}
\end{eqnarray}
where the rational factors in $w_i$'s should be understood in 
the sense of the series as
$(1-w_j/w_i)/( 1-q^{\pm 1}w_j/w_i)=1+(1-q^{\mp 1})\sum_{n>0} (q^{\pm 1}w_j/w_i)^n$.
\end{definition}

For example, 
we have  $I_1=\eta_0$ and $I_2=\eta_0^2+(1-q^{-1})\sum_{n>0} q^n \eta_{-n}\eta_{n}$, 
and so on. Based on the argument given in Ref. \refcite{FHHSY}
with considering the classical limit ($t\rightarrow 1$), 
one can prove the following.
\begin{proposition}
We have the commutativity
$ \{I_k,I_l\}=0$, $ \{\overline{I}_k,\overline{I}_l\}=0$
and $ \{I_k,\overline{I}_l\}=0$.
\end{proposition}

\subsection{Integrals of motion associated with $t$ and $\overline{t}$}
Some explicit calculations show us that the 
integrals $I_k$ and $\overline{I}_k$ does not correspond to the 
Toda times $t_1,t_2,\cdots$ and $\overline{t}_1,\overline{t}_2,\cdots$ 
in general. Hence our task is 
to find a suitable set of integrals, which we call  $M_k$ and $\overline{M}_k$.

At present, unfortunately, it is not easy to do the task purely within the 
framework of Poisson algebra.
However, with the knowledge of the
values of the integrals $I_k$ and $\overline{I}_k$ on 
the $n$-soliton solution  (\ref{tau-plus}), (\ref{tau-minus}),
we can guess the correct formula.

\begin{conjecture}\label{con1}
Let $\tau_\pm(z,t,\overline{t})$ be as in (\ref{tau-plus}) and (\ref{tau-minus}).
The quantities $I_k$'s and $\overline{I}_k$'s are independent of 
$t$ and $\overline{t}$. The values are 
given by the follwing specialization of 
the elementary symmetric functions as
\begin{eqnarray}
&&I_k=q^{-k(k-1)/2}(1-q)(1-q^2)\cdots(1-q^k)\label{I_k}\\
&&\qquad \times e_k(a_1,\cdots,a_n,q^n \varepsilon,q^{n+1} \varepsilon,\cdots),\nn\\
&&\overline{I}_k=q^{k(k-1)/2}(1-q^{-1})(1-q^{-2})\cdots(1-q^{-k})\\
&&\qquad \times e_k(a_1^{-1},\cdots,a_n^{-1},q^{-n} \varepsilon^{-1},q^{-n-1} \varepsilon^{-1},\cdots).\nn
\end{eqnarray}
\end{conjecture}

As for the statement about $I_k$ in Eq. (\ref{I_k}), see Ref. \refcite{TS}.

\begin{remark}
For small $k$, we have
\begin{eqnarray*}
&&I_1=(1-q)(a_1+\cdots+a_n)+q^n \varepsilon,\\
&&I_2=q^{-1}(1-q)(1-q^2)(a_1a_2+a_1 a_3+\cdots+a_{n-1}a_n)\\
&&\qquad +q^{n-1}(1-q^2)(a_1+\cdots+a_n) \varepsilon+q^{2n} \varepsilon^2,
\end{eqnarray*}
and so on.
\end{remark}

\begin{definition}\label{MMbar}
Set $I'_k=q^{k(k-1)/2}((1-q)(1-q^2)\cdots (1-q^k))^{-1}I_k$ 
and $\overline{I}'_k=q^{-k(k-1)/2}((1-q^{-1})(1-q^{-2})\cdots (1-q^{-k}))^{-1}\overline{I}_k$ 
for simplicity of display.
Define
\begin{eqnarray}
&&M_k={1-q^k\over k}
\left| 
\begin{array}{lllll} 
I'_1&1&0&\cdots&0\\
2I'_2&I'_1&1&\cdots&\\
3I'_3&I'_2&I'_1&\cdots&\\
\vdots&&&\ddots&\vdots\\
kI'_k&I'_{k-1}&I'_{k-2}&\cdots&I'_1
\end{array}\right|,\\
&& 
\overline{M}_k={1-q^{-k}\over k}
\left| 
\begin{array}{lllll} 
\overline{I}'_1&1&0&\cdots&0\\
2\overline{I}'_2&\overline{I}'_1&1&\cdots&\\
3\overline{I}'_3&\overline{I}'_2&\overline{I}'_1&\cdots&\\
\vdots&&&\ddots&\vdots\\
k\overline{I}'_k&\overline{I}'_{k-1}&\overline{I}'_{k-2}&\cdots&\overline{I}'_1
\end{array}\right|.
\end{eqnarray}
\end{definition}

\begin{remark}
Conjecture \ref{con1} implies that 
if $\tau_\pm(z,t,\overline{t})$ be as in (\ref{tau-plus}) and (\ref{tau-minus}),
$M_k$'s and $\overline{M}_k$'s are given by
\begin{eqnarray}
&&M_k={1-q^k\over k}p_k(a_1,\cdots,a_n,q^n \varepsilon,q^{n+1} \varepsilon,\cdots),\nn\\
&&\overline{M}_k={1-q^{-k}\over k}p_k(a_1^{-1},\cdots,a_n^{-1},q^{-n} \varepsilon^{-1},q^{-n-1} \varepsilon^{-1},\cdots).\nn
\end{eqnarray}
\end{remark}

\subsection{Formulas for $M_2$ and $M_3$}
Now we come back to our study of the 
Poisson algebra to check the  higher Hirota equations (\ref{To-2}) and (\ref{To-3}).

It is desirable to find some reasonably simple expressions
for $M_k$'s.
At present, however, we only have the following partial results.
\begin{lemma}
We have $M_1=[\eta(w)]_1$ and 
\begin{eqnarray}
&&M_2=\left[\left({1\over 2}+{q w_2/w_1\over 1-q w_2/w_1}\right)
\eta(w_1)\eta(w_2) \right]_1,\label{M_2}\\
&&M_3=\left[\left({1\over 3}+{q w_3/w_2\over (1-q w_2/w_1)(1-q w_3/w_2)}\right)
\eta(w_1)\eta(w_2)\eta(w_3) \right]_1.\label{M_3}
\end{eqnarray}
\end{lemma}

\begin{remark}
It is an open problem to find a simple expression as above
for $M_4,M_5,\cdots$.
\end{remark}

\subsection{Main conjecture and equations with respect to $t_2$, $t_3$}

\begin{definition}\label{PBtime}
Set $\partial_{t_k}*=\{M_k,*\}$ and $\partial_{\overline{t}_k}*=\{*,\overline{M}_k\}$.
\end{definition}

Now we are ready to state our conjecture.
\begin{conjecture}\label{maincon}
Calculating $\partial_{t_k}$ and $\partial_{\overline{t}_k}$ by using
the Poisson brackets given in Definition \ref{PBtime},
we recover the same equation derived from 
the Hirota-Miwa equations (\ref{HM+-1}), (\ref{HM+-2}) and (\ref{HM3}).
\end{conjecture}

The rest of the paper is devoted to give some evidence of 
our conjecture.

\begin{proposition}
Calculating $\partial_{t_2}$ and $\partial_{t_1}$ by using
the Poisson brackets given in Definition \ref{PBtime},
we recover the same equation as in Eq. (\ref{To-2}).
\end{proposition}
\begin{proof}
{}From (\ref{eta-tau-}), (\ref{eta-tau+}) and (\ref{M_2}), we have
\bb
&&{\partial_{t_2} \tau_-(z)\over \tau_-(z)}-
{\partial_{t_2} \tau_+(z)\over \tau_+(z)}+2 M_2\nn\\
&&
=\left[
\left(\delta(w_1/z)+\delta(w_2/z)\right)
\left({1\over 2}+\frac{q w_2/w_1}{1-qw_2/w_1}\right)\eta(w_1)\eta(w_2)\right]_{1,w_1,w_2}\nn\\
&&=M_1\eta(z)+
\eta(z)
\left[\frac{qw_2/z}{1-qw_2/z}\eta(w_2)\right]_{1,w_2}
+
\eta(z)
\left[\frac{qz/w_1}{1-qz/w_1}\eta(w_1)\right]_{1,w_1}
\nn\\
&&=M_1\eta(z)+
\eta(z)\eta_-(z/q)+
\eta(z)
\eta_+(zq).\nn
\nn\ee
Using (\ref{eta}) and (\ref{partial-tau}) (with $t=t_1$), we have the result.
\end{proof}

Finally, we study the Hirota equation involving the third time $t_3$.
\begin{proposition}\label{t_3}
Calculating $\partial_{t_3}$, $\partial_{t_2}$ and $\partial_{t_1}$ by using
the Poisson brackets given in Definition \ref{PBtime}, we have
\bb
&&\left(D_{t_3}+3 M_3\right)\tau_-(z)\cdot \tau_+(z)\\
&&=
{1\over 2}\varepsilon \left(D_{t_2}+2M_2\right) \tau_-(z/q)\cdot \tau_+(zq)
+{1\over 2}\varepsilon \left( D_{t_1}+M_1\right)^2 \tau_-(z/q)\cdot \tau_+(zq),\nn\\
&&\left(D_{t_1}+ M_1\right)^3\tau_-(z)\cdot \tau_+(z)\\
&&=
2\varepsilon \left(D_{t_2}+2M_2\right) \tau_-(z/q)\cdot \tau_+(zq)
-\varepsilon \left( D_{t_1}+M_1\right)^2 \tau_-(z/q)\cdot \tau_+(zq).\nn
\ee
\end{proposition}

\begin{corollary}
We recover (\ref{To-3}) from the the Hamiltonian structure.
\end{corollary}

\begin{proof}[Proof of Proposition \ref{t_3}]
They follow from Lemmas \ref{Lemma-1}, \ref{Lemma-2},\ref{Lemma-3} and \ref{Lemma-4}
below.
\end{proof}

\begin{lemma}\label{Lemma-1}
{}From (\ref{eta-tau-}), (\ref{eta-tau+}) and (\ref{M_3}), we have
\begin{eqnarray*}
&&{\left(D_{t_3}+3M_3\right) \tau_-(z)\cdot \tau_+(z)\over  \tau_-(z)\cdot \tau_+(z)}\\
&&=\eta(z)\left(M_2+{1\over 2}M_1^2+M_1(\eta_+(zq)+\eta_-(z/q))+\eta_+(zq)\eta_-(z/q)\right)\\
&&+\eta(z)\left[
\left({qw_1/w_2\over 1-q w_1/w_2}{qw_2/z\over 1-q w_2/z}\right.\right.\\
&&\qquad\qquad \left.\left.
+{qw_2/w_1\over 1-q w_2/w_1}{qz/w_2\over 1-qz/ w_2}\right)
\eta(w_1)\eta(w_2)\right]_{1,w_1,w_2}.
\nn\end{eqnarray*}
\end{lemma}

\begin{lemma}\label{Lemma-2}
{}From  (\ref{eta-eta}), (\ref{eta-tau-}) and (\ref{eta-tau+}), we have
\begin{eqnarray*}
&&{\left(D_{t_1}+M_1\right)^3 \tau_-(z)\cdot \tau_+(z)\over  \tau_-(z)\cdot \tau_+(z)}\\
&&=\eta(z)\Biggl(4M_2-M_1^2+M_1(\eta_+(zq)+\eta_-(z/q))-2\eta_+(zq)\eta_-(z/q)\Biggr)\\
&&+2\eta(z)\left[
\left({qw_1/w_2\over 1-q w_1/w_2}+{qw_2/w_1\over 1-q w_2/w_1}\right)\right.\\
&&\qquad\qquad\left.\times
\left(
{qw_2/z\over 1-q w_2/z}+{qz/w_2\over 1-qz/ w_2}\right)
\eta(w_1)\eta(w_2)\right]_{1,w_1,w_2}\\
&&-\eta(z)\left[
\left({qw_1/w_2\over 1-q w_1/w_2}-{qw_2/w_1\over 1-q w_2/w_1}\right)\right.\\
&&\qquad\qquad\left.\times
\left(
{qw_2/z\over 1-q w_2/z}-{qz/w_2\over 1-qz/ w_2}\right)
\eta(w_1)\eta(w_2)\right]_{1,w_1,w_2}.
\nn\end{eqnarray*}
\end{lemma}

\begin{lemma}\label{Lemma-3}
{}From (\ref{eta-tau-}), (\ref{eta-tau+}) and (\ref{M_2}), we have
\begin{eqnarray*}
&&{\left(D_{t_2}+2M_2\right) \tau_-(z/q)\cdot \tau_+(zq)\over  \tau_-(z/q)\cdot \tau_+(zq)}\\
&&=2M_2+M_1(\eta_+(zq)+\eta_-(z/q))\\
&&+\left[
\left({qw_1/w_2\over 1-q w_1/w_2}+{qw_2/w_1\over 1-q w_2/w_1}\right)\right.\\
&&\qquad\qquad\left.\times
\left(
{qw_2/z\over 1-q w_2/z}+{qz/w_2\over 1-qz/ w_2}\right)
\eta(w_1)\eta(w_2)\right]_{1,w_1,w_2}.
\nn\end{eqnarray*}
\end{lemma}

\begin{lemma}\label{Lemma-4}
{}From  (\ref{eta-eta}), (\ref{eta-tau-}) and (\ref{eta-tau+}), we have
\begin{eqnarray*}
&&{\left(D_{t_1}+M_1\right)^2 \tau_-(z/q)\cdot \tau_+(zq)\over  \tau_-(z/q)\cdot \tau_+(zq)}\\
&&=M_1^2+M_1(\eta_+(zq)+\eta_-(z/q))+2\eta_+(zq)\eta_-(z/q)\\
&&+\left[
\left({qw_1/w_2\over 1-q w_1/w_2}-{qw_2/w_1\over 1-q w_2/w_1}\right)\right.\\
&&\qquad\qquad\left.\times
\left(
{qw_2/z\over 1-q w_2/z}-{qz/w_2\over 1-qz/ w_2}\right)
\eta(w_1)\eta(w_2)\right]_{1,w_1,w_2}.
\nn\end{eqnarray*}
\end{lemma}

\section*{Acknowledgments}
The authors thank Kenji Kajiwara, Paul Wiegmann and  Masatoshi Noumi for 
stimulating discussions. YT is supprted by The Promotion and 
Mutual Aid Corporation for Private Schools of Japan.

\end{document}